\documentclass[conference]{IEEEtran}

\IEEEoverridecommandlockouts
\usepackage[english]{babel}
\usepackage{amssymb,amsmath,amscd,amsfonts,amsthm,bbm,bm}
\usepackage{amsmath}
\usepackage{graphicx}
\usepackage{psfrag}
\usepackage{paralist}
\usepackage{color}
\usepackage{float}
\usepackage{relsize}
\usepackage{euscript}
\usepackage{setspace}
\usepackage{flushend}
\usepackage{tikz}
\usetikzlibrary{plotmarks}
\usepackage{pgfplots}
\usetikzlibrary{arrows,shapes,decorations,backgrounds}
\usepackage{calc}
\pgfplotsset {compat = 1.16}
\usepackage[nolist,printonlyused]{acronym}      
\usepackage{comment}
\usepackage{cite}

\def\BibTeX{{\rm B\kern-.05em{\sc i\kern-.025em b}\kern-.08em
		T\kern-.1667em\lower.7ex\hbox{E}\kern-.125emX}}

\DeclareMathOperator{\varr}{var}

\DeclareMathOperator{\diag}{diag}

\newcommand{\norme}[1]{\left\Vert #1\right\Vert}

\newcommand{\R}{\mathbb R}
\newcommand{\Z}{\mathbb Z}
\newcommand{\C}{\mathbb{C}}
\newcommand{\E}{\mathbb{E}}

\newcommand{\br}{\}}

\def\bx{{\bf x}}\def\bx{{\bf x}}
\def\bX{{\bf X}}
\def\bR{{\bf R}}
\def\bA{{\bf A}}
\def\bI{{\bf I}}

\def\bY{{\bf Y}}
\def\bH{{\bf H}}

\def\bQ{{\bf Q}}

\def\bZ{{\bf Z}}

\def\ba{{\bf a}}

\def\bh{{\bf h}}
\def\bw{{\bf w}}
\def\bW{{\bf W}}
\def\bS{{\bf S}}
\def\bZ{{\bf Z}}

\def\br{{\bf r}}

\def\by{{\bf y}}
\def\bz{{\bf z}}
\def\bq{{\bf q}}

\def\bTT{{\bf {\cal T}}}

\newtheorem{lemma}{Lemma}

\newtheorem{theorem}{Theorem}

\begin{document}

\title{On Estimating the Autoregressive Coefficients of Time-Varying Fading Channels}

\author{
	Julia Vinogradova$^{\star}$,
	G\'{a}bor Fodor$^{\ddag \dag }$,
	Peter Hammarberg$^{\ddag}$, \\
	
	\small $^{\star}$Ericsson Research, Finland, Email: \texttt{Julia.Vinogradova@ericsson.com}\\	
	\small $^{\ddag}$Ericsson Research, Sweden, E-mail: \texttt{firstname.secondname@ericsson.com} \\
	\small $^\dag$KTH Royal Institute of Technology, Sweden. E-mail: \texttt{gaborf@kth.se}\\
	
}

\maketitle

\begin{acronym}[LTE-Advanced]
  \acro{2G}{Second Generation}
  \acro{3G}{3$^\text{rd}$~Generation}
  \acro{3GPP}{3$\text{rd}$~Generation Partnership Project}
  \acro{4G}{4$^\text{th}$~Generation}
  \acro{5G}{5$^\text{th}$~Generation}

  \acro{AR}{autoregressive}

  \acro{CDF}{cumulative distribution function}

  \acro{MIMO}{multiple-input multiple-output}
    \acro{SIMO}{single-input multiple-output}

  \acro{SNR}{signal-to-noise ratio}
 
\end{acronym}

\begin{abstract}
As several previous works have pointed out, the evolution of the wireless channels in
multiple input multiple output systems can be advantageously modeled as an autoregressive process. Therefore, estimating the coefficients, and, in particular, the state transition matrix of this autoregressive process is a key to accurate channel estimation, tracking, and prediction in fast fading environments.
In this paper we assume the time varying spatially uncorrelated channel which is approximately the case with proper antenna spacing at the base station in rich scattering environments. We propose a method for autoregressive parameter estimation for a single input multiple output (SIMO) channel. We show an almost sure convergence of the estimated coefficients to the true autoregressive coefficients in large dimensions. We apply the proposed method to the SIMO channel tracking.
\end{abstract}


\IEEEpeerreviewmaketitle
\section{Introduction}

It is well-known that the temporal variations of wireless channels due to changes in the propagation
environment or mobility are advantageously modeled by \ac{AR} processes. When the parameters of the
\ac{AR} process are accurately estimated, estimating and predicting the process states, and thereby
the wireless channel coefficients become feasible by Kalman filters \cite{Yan:01, Hijazi:10}
Therefore, a large body of works related to
the estimation of the parameters of \ac{AR} processes as well as the application of such processes
to channel estimation, prediction, equalization and detection exists
\cite{Zhang:07B, Lehmann:08, Abeida:10, Hijazi:10, Truong:13, Kong:15, Zhang:16, Papa:18, Kim:20}.
Approximating mobile channels in single input single output systems with an \ac{AR} model has been studied
in e.g. \cite{Zhang:07B} and \cite{Abeida:10}. 
While in \cite{Abeida:10} a first order \ac{AR}
model is used for data-aided \ac{SNR} estimation, the works reported in \cite{Hijazi:10} and \cite{Zhang:07B}
use higher order \ac{AR} models for developing channel estimation and data detection algorithms. In contrast,
papers \cite{Yan:01} and \cite{Lehmann:08} study \ac{MIMO} systems in fast Rayleigh fading environments and
use \ac{AR} processes to characterize the temporal variations of the channels, and evaluate their effects on
the receiver structures and performance. More recently, paper \cite{Zhang:16} developed algorithms for
tracking the angles of departures and arrivals in multi-antenna systems using extended Kalman filter.

In the context of large-scale \ac{MIMO} systems, a series of recent works have focused on combatting
the negative effects of channel aging \cite{Truong:13, Kong:15, Papa:18, Kim:20}. These papers also make use
of the characteristics of \ac{AR} models for channel estimation and prediction purposes, since high quality
channel state information is needed for various \ac{MIMO} algorithms, including data reception in the
uplink and spatial precoding in the downlink. Recognizing the importance of properly mapping the \ac{AR}
process to the measured wireless channel variations, papers \cite{Kim:20, Yuan:20} use the Yule-Walker
and Levinson-Durbin equations to identify the \ac{AR} system parameters.
In a recent work reported in \cite{Esfandiari:20} several algorithms to estimate the \ac{AR} coefficients of $p$-order processes, denoted by AR($p$), are developed. The proposed algorithms in \cite{Esfandiari:20} are useful in practice, because they not only estimate the \ac{AR}
coefficients, but also the variance of the observation and process noise, based only on measurements that are feasible in practice.

In this paper, we propose an estimation method of \ac{AR} coefficients of a \ac{SIMO} channel vector following AR($p$) process making use of measurements over the time and spatial dimensions. We assume the signals arriving at multiple antennas are uncorrelated and the variances or the underlying noise processes can be estimated by existing noise variance estimation schemes \cite{Esfandiari:20}. The main contribution consists in deriving concentration inequalities that are useful for evaluating the consistency of the proposed estimators. To illustrate the operation of the proposed technique, we apply it to channel tracking in a \ac{SIMO} system.

The rest of this paper is structured as follows. The next section presents our system model.
Next, Section \ref{Sec:Est} proposes an estimation of the AR($p$) parameters, while
Section \ref{Sec:App} shows a specific application of the \ac{AR} model in the context of wireless
channel tracking. Section \ref{Sec:Sim} discusses numerical results obtained by the proposed
estimation scheme and compares the results to relevant benchmarks. Finally, Section \ref{Sec:Con}
concludes the paper.

\section{System Model}
\label{Sec:Sys}
\subsection{\ac{AR} $p$-deep model}
We consider the vector $\bh(t)=[h_0(t), \ldots, h_{N_r-1}(t)]^{\sf T}\in \C^{N_r \times 1}$ with independent and identically distributed (i.i.d.) elements for $n \in {0, \ldots N_r-1}$. We assume each element $\left(h_n(t)\right)_{t \in \Z}$, for $n \in \left\{0, \ldots, N_r-1 \right\}$, is a complex Gaussian stationary process following the model AR($p$) in time. Assuming the size of the observation window is equal to $T$, at time instant $t \in {0, \ldots, T-1}$, we have
\begin{align}\label{transition}
\bh(t) &= \bA_1 \bh(t-1) + \dots + \bA_p \bh(t-p) + \bx(t)
\end{align}
where the matrices $\bA_i \in \C^{N_r \times N_r}$ for $i \in \left\{1, \ldots, p \right\}$ are assumed to be constant in time, and $\bx(t)=[x_0(t), \ldots, x_{N_r-1}(t)]^{\sf T} \in \C^{N_r \times 1}$ is the process white noise with i.i.d. elements $x_{n}(t) \sim {\cal CN}(0, \sigma_x^2)$ where the notation ${\cal CN}(0, \sigma^2)$ represents the complex circular Gaussian distributions with mean $a$ and variance $\sigma^2$. Note that as $\bh$ has i.i.d. elements in space, the matrices $\bA_i$ are diagonal and equal to $\bA_i=a_i\bI_{N_r}$ for $i \in \left\{1, \ldots, p \right\}$ with $a_i$ referring to the \ac{AR} coefficients and the notation $\bI_{n}$ denoting the identity matrix of dimension $n \times n$. In the time-variant system framework (such as Kalman filter), Equation \eqref{transition} refers to the state transition equation, and the matrices $\bA_i$ refer to the state transition matrices and are usually assumed to be known. However, in realistic scenarios, $\bA_i$ are not known and need to be estimated. In time-variant systems, it is assumed that there is an observation model and we consider the observation equation at time instant $t \in {0, \ldots, T-1}$ is given by
\begin{equation}\label{observation}
{\hat \bh}(t) = \bh(t) + \bw(t)
\end{equation}
where $\bw(t)=[w_0(t), \ldots, w_{N_r-1}(t)]^{\sf T} \in \C^{N_r \times 1}$ is the observation white noise with i.i.d. elements $w_{n}(t) \sim {\cal CN}(0, \sigma_w^2)$.

In the following, we propose a method for estimating the AR coefficients $a_i$ of the above model for $i \in \left\{1, \ldots, p \right\}$ from the $T$ observations given in \eqref{observation}. Before presenting the proposed method, we need to describe the time covariance matrix (including the covariance coefficients) of the vector $\hat{\bh}$, the estimation of which is the main step of the proposed algorithm.
\par In the following, we assume that $N_r$ and $T$ are of the same order and converge to the infinity such that $N_r/T \to c>0$. In practice, as we will see in the simulation part, $N_r$ and $T$ can take finite values in order to achieve a reasonable performance.
\vspace{-0.15cm}

\subsection{Covariance matrix of $\hat{\bh}$}
\label{cov_h_hat}

The covariance function of $h_n(t)$ process, for $k=1-T, \ldots, T-1$, is defined as
\begin{equation*}
r(k) \triangleq \E\left[h_n(t)h_n(t-k)^*\right]
\end{equation*}
for each $n$th element of the vector $\bh(t) \in \C^{N_r \times 1}$. We assume in this paper the absolute summability of the covariance coefficients resulting in a bounded sum $\sum_{k=1-T}^{T-1}|r(k)| \leq K$ where $K$ is a positive fixed constant as $T,N_r \to \infty$. This assumption is not restrictive in general and holds in a large variety of the practical cases.

The covariance matrix of $\bh_n=[h_n(0), \ldots, h_n(T-1)]\in \C^{1 \times T}$ is given by
\begin{align*}
&\bR \triangleq \E\left[\bh_n^{\sf H}\bh_n\right]=\bTT \left({r}(1-T), \ldots, {r}(T-1)\right)
\end{align*}
where $\bTT \left({r}(1-T), \ldots, {r}(T-1)\right)$ refers to the Toeplitz matrix \cite{Gray'06} formed from the coefficients ${r}(1-T), \ldots, {r}(T-1)$.

The covariance matrix of the observation model $\hat \bh_n=[{\hat h}_n(0), \ldots, {\hat h}_n(T-1)]\in \C^{1 \times T}$, with using \eqref{observation}, is given by:

\begin{align}
\bR_{\hat \bh}& \triangleq \E\left[{\hat \bh_n}^{\sf H}{\hat \bh_n}\right]=\E\left[\left(\bh_n+\bw_n\right)^{\sf H}\left(\bh_n+\bw_n\right)\right] \nonumber\\
&=\E\left[\bh_n^{\sf H} \bh_n\right]+\E\left[\bw_n ^{\sf H}\bw_n\right]=\sigma_x^2\bR+\sigma_w^2\bI_{T} \label{Rhat}
\end{align}
where  $\bw_n=[w_n(0), \ldots, w_n(T-1)] \in \C^{1 \times T}$.

\vspace{-0.15cm}

\subsection{Definition of the \ac{AR} coefficients}
For $k=1, \ldots, p$, the Yule-Walker equations \cite{Kay'1988} are given by
\begin{equation*}
r(k) = \sum_{i=1}^{p} a_i r(k-i).
\end{equation*}

From these equations and defining $\bR_{p} \triangleq \bTT \left({r}(1-p), \ldots, {r}(p-1)\right)$ as a $p$-truncated version of $\bR$ for any $p \in \left\{1,\ldots, T\right\}$, we can write a linear system of equations in the matrix form as:
\begin{align*}
{\br}_p={\bR}_{p} \ba_p
\end{align*}
where ${\br}_p=[{r}(1), \ldots, {r}(p)]^{\sf T}\in \R^{p \times 1}$ and $\ba_p=[a_1, a_2, \ldots, a_p]^{\sf T}\in \R^{p \times 1}$.

As the matrix ${\bR}_{p}$ is of full rank, it is invertible and we can express ${\ba}_p$ as:
\begin{align*}
{\ba}_p={\bR}_{p}^{-1}{\br}_p.
\end{align*}

The remaining of this paper deals with the estimation of the vector of \ac{AR} coefficients ${\ba}_p$.

\section{Estimation of the \ac{AR} Coefficients}
\label{Sec:Est}

The estimation of the \ac{AR} coefficient vector ${\ba}_p$ is based on the estimation of ${\br}_p$ and ${\bR}_{p}$ whose estimates are provided in the following two subsections. The main result is presented in the third subsection.

\subsection{Estimation of ${\br}_{p}$}

Concatenating the observation vector $\hat{\bh}(t) \in \C^{N_r \times 1}$ from \eqref{observation} over $T$ observations, we can write
\begin{equation}\label{H_hat}
\widehat{\bH}=\left[\hat{h}_n(t)\right]_{n,t=0}^{N_r-1, T-1}=[\hat{\bh}(0), \ldots, \hat{\bh}(T-1))]=\bX\bR^{1/2}+\bW
\end{equation}
where $\bX=[\bx(0), \ldots, \bx(T-1) ] \in \C^{N_r \times T}$ with $x_{n}(t) \sim {\cal CN}(0, \sigma_x^2)$ and $\bW=[\bw(0), \ldots, \bw(T-1) ] \in \C^{N_r \times T}$ with $w_{n}(t) \sim {\cal CN}(0, \sigma_w^2)$, as previously defined.

The following Lemma, which is an adapted version of the estimates proposed in \cite{Vino'2014}, provides the estimates for the $r(k)$ coefficients.

\begin{lemma}\label{Lemma_Vino_rk}
	Let the observation matrix $\widehat{\bH}=\left[\hat{h}_n(t)\right]_{n,t=0}^{N_r-1, T-1}$ be defined as in \eqref{H_hat}.
	The biased and unbiased estimates of $r(k)$ are given, respectively, for $k=0$ by
	\begin{align*}
	\hat{r}^{b,u}(0) &= \frac{1}{\sigma_x^2N_rT}\sum_{n=0}^{N_r-1}\sum_{t=0}^{T-1} \hat{h}_n(t) \hat{h}_n(t)^* -\frac{\sigma_w^2}{\sigma_x^2}
	\end{align*}
	and for $k=1-T, \ldots, -1$ and $k=1, \ldots, T-1$ by
	\begin{align*}
	\hat{r}^b(k) &= \frac{1}{\sigma_x^2N_rT}\sum_{n=0}^{N_r-1}\sum_{t=0}^{T-1} \hat{h}_n(t+k) \hat{h}_n(t)^* \\
	\hat{r}^u(k) &= \frac{1}{\sigma_x^2N_r(T-|k|)} \sum_{n=0}^{N_r-1} \sum_{t=0}^{T-1} \hat{h}_n(t+k) \hat{h}_n(t)^*
	\end{align*}
	for $0 \leq t+k \leq T-1$.
Then, for $k=1-T, \ldots, T-1$, for any $\epsilon >0$, we have
\begin{align*}
\mathbb{P} \left[\left|\hat{r}^b(k) - r(k) \right| \geq \epsilon \right] &\leq \frac{K'}{\epsilon^2N_rT}\\
\mathbb{P} \left[\left|\hat{r}^u(k) - r(k) \right| \geq \epsilon \right] & \leq \frac{K'}{\epsilon^2N_r(T-|k|)}.
\end{align*}
where $K'>0$ is a positive constant.
\end{lemma}
\begin{proof}
	The proof is provided in Appendix~\ref{Proof_lemma1}.
\end{proof}
We notice (see the proof), that the variances of the errors above estimators converge to zero with the rates $1/(N_rT)$ for the biased case and $1/(N_r(T-|k|))$ for the unbiased case.

\subsection{Estimation of ${\bR}_{p}$}
We now provide the estimates of the covariance matrix based on the estimated coefficients and the results from \cite{Vino'2014}.
\begin{lemma}\label{Lemma_Vino_R}
	Let, for $k=1-T, \ldots, T-1$, $\hat{r}(k)^b$ and $\hat{r}(k)^u$ be the biased and unbiased estimates of $r(k)$, respectively, defined as in Lemma \ref{Lemma_Vino_rk}.
	Define the estimated covariance matrices as
	\begin{align*}
	\widehat \bR^b_p & \triangleq \bTT \left(\hat{r}^b(1-p), \ldots, \hat{r}^b(p-1)\right)\\
	\widehat \bR^u_p & \triangleq \bTT \left(\hat{r}^u(1-p), \ldots, \hat{r}^u(p-1)\right).
	\end{align*}

Then, for any $\epsilon >0$, we have
\begin{align*}
\mathbb{P} \left[\norme{\widehat \bR^b_p - \bR_p} > \epsilon \right]& \leq
\exp \left( -cT \left( \frac{\epsilon}{C} -
\log \left( 1 + \frac{\epsilon}{C} \right) + o(1) \right)
\right) \\
\mathbb{P} \left[\norme{\widehat \bR^u_p - \bR_p} > \epsilon \right]
&\leq
\exp \left(- \frac{C'T\epsilon^2}{\log T} (1 + o(1)) \right)
\end{align*}
where $o(1)$ is with respect to $T$ and depends on $\epsilon$ and $c$, $C$, and $C'$ are positive and bounded as $T\to \infty$, and $\norme{\cdot}$ denotes the spectral norm.
\end{lemma}
\begin{proof}
The proof is provided in Appendix~\ref{Proof_lemma2}.
\end{proof}

\subsection{Estimation of the \ac{AR} coefficients ${\ba}_{p}$: main result}
Based on the results of the above subsections, we can now define the estimates of the AR coefficients in the following theorem.
\begin{theorem}
	Let $\hat{r}^{b,u}(k)$ be the biased or unbiased estimate defined as in Lemma~1. Define $\widehat{\br}^{b,u}_p \triangleq [\hat{r}^{b,u}(1), \ldots, \hat{r}^{b,u}(p)]^{\sf T}\in \R^{p \times 1}$. We define
	\begin{align*}
	\widehat {\bR^{b,u}}_{p} & \triangleq \bTT \left(\hat{r}^{b,u}(1-p), \ldots, \hat{r}^{b,u}(p-1)\right).
	\end{align*}
	The biased and unbiased estimators are given by
	\begin{align*}
	\hat{\ba}^{b,u}_p = {\widehat{\bR^{b,u}}_{p}}^{-1}\hat{\br}^{b,u}_p
	\end{align*}
	where $\hat{\ba}^{b,u}_p=[\hat{a}^{b,u}_1, \ldots, \hat{a}^{b,u}_p]^{\sf T}$.
	Then, for $i=1, \ldots, p$, for any $\epsilon >0$
\begin{align*}
\mathbb{P} \left[\left|\hat{a}^b_i-{a}_i\right| > \epsilon \right] & \leq \frac{K''}{\epsilon^2N_rT}\\
\mathbb{P} \left[\left|\hat{a}^u_i-{a}_i\right| > \epsilon \right] & \leq \frac{K''}{\epsilon^2N_r(T-|i|)} .
\end{align*}
where $K''>0$ is a positive constant.
\end{theorem}
From this theorem we have the almost sure convergence of the proposed estimator of $a_i$ to the true value for all $i=1, \ldots, p$.

\begin{proof}
	The proof is provided in Appendix~\ref{Proof_theorem1}.
\end{proof}
\section{Application to Channel Tracking}
\label{Sec:App}
We consider a communication system with $N_t=1$ transmit antennas and $N_r$ receive antennas.
We assume that the random channel coefficients $\left[{h}_n(t)\right]_{n,t=0}^{N_r-1, T-1}$ are i.i.d. in space, follow the AR($p$) in time. Assuming an uplink transmission, the ${N_r \times 1}$ received signal at the base station is given by
\begin{equation*}
\by(t)=\bh(t)s(t)+\bw(t)=\sum_{k=1}^p a_k\bh(t-k)s(t)+\bw(t)
\end{equation*}
where $s(t) \in \C$ is the transmitted signal, $\bw(t)$ is the noise vector defined as in \eqref{observation}, and $a_k$ for $k \in \left\{1, \ldots, p \right\}$ are the coefficients of the AR($p$) process model.

Concatenating over $T$ time slots, the ${N_r \times T}$ received signal matrix is given by
\begin{equation}\label{received_mtx}
\bY=[\by(0), \ldots, \by(T-1)]=\bH\bS+\bW
\end{equation}
where $\bS=\diag \left\{s(0),\ldots,s(T-1)\right\}$ is a unitary matrix such that $\bS\bS^{\sf H}=\bI_{T}$.
With this assumption, it is clear that the covariance matrix in time of the received signal is equal to the covariance matrix of $\hat{\bh}$ from Section \ref{cov_h_hat}. Hence, Theorem~1 can be directly applied with the observation model \eqref{received_mtx} in order to get the AR estimates $\hat{a}_k$ for $k=1,\ldots,p$. The channel tracking equation at time $t$ is then given by
\begin{equation*}
\hat{\bh}(t) = \sum_{i=1}^p \hat{a}_i\hat{\bh}(t-p)
\end{equation*}
where $\hat{a}_i$ are obtained from Theorem~1.

\section{Simulation Results}
\label{Sec:Sim}

\subsection{\ac{AR} coefficient estimation}
In this subsection, we show the performance of the proposed biased and unbiased estimators from Theorem~1 and compare it with the performance of an existing estimator referred here to the time-based method. The time-based estimator is similar to the one given in \cite{Esfandiari:20}, for which there is no averaging over the spatial domain, {i.e.}, assuming $N_r=1$.
We consider the AR($2$) for which $a_1 \in [0,2)$ and $a_2\in (-1,0)$.
We assume here that $a_1=1.8$ and $a_2=-0.9$.
The choice of these values is motivated by the commonly used channel model corresponding to the so called Jakes' model \cite{Baddour:05}.
It has been shown in \cite{ElHusseini'2019}, that in order to approximate the Jakes' Doppler spectrum, the coefficients $a_1$ and $a_2$ should be necessarily close to $2$ and $-1$, respectively.
In Figure~1 the normalized mean square errors (NMSE) of the estimates versus $T$ (assumed to be equal to $N_r$) are depicted.
We observe an important improvement in performance of the proposed estimators as compared to the time-based estimator which is due to the averaging over $N_r$ samples in the spatial domain.
\begin{figure}[ht]\label{AR_coef_simu}
		\center
	\includegraphics[width=7cm]{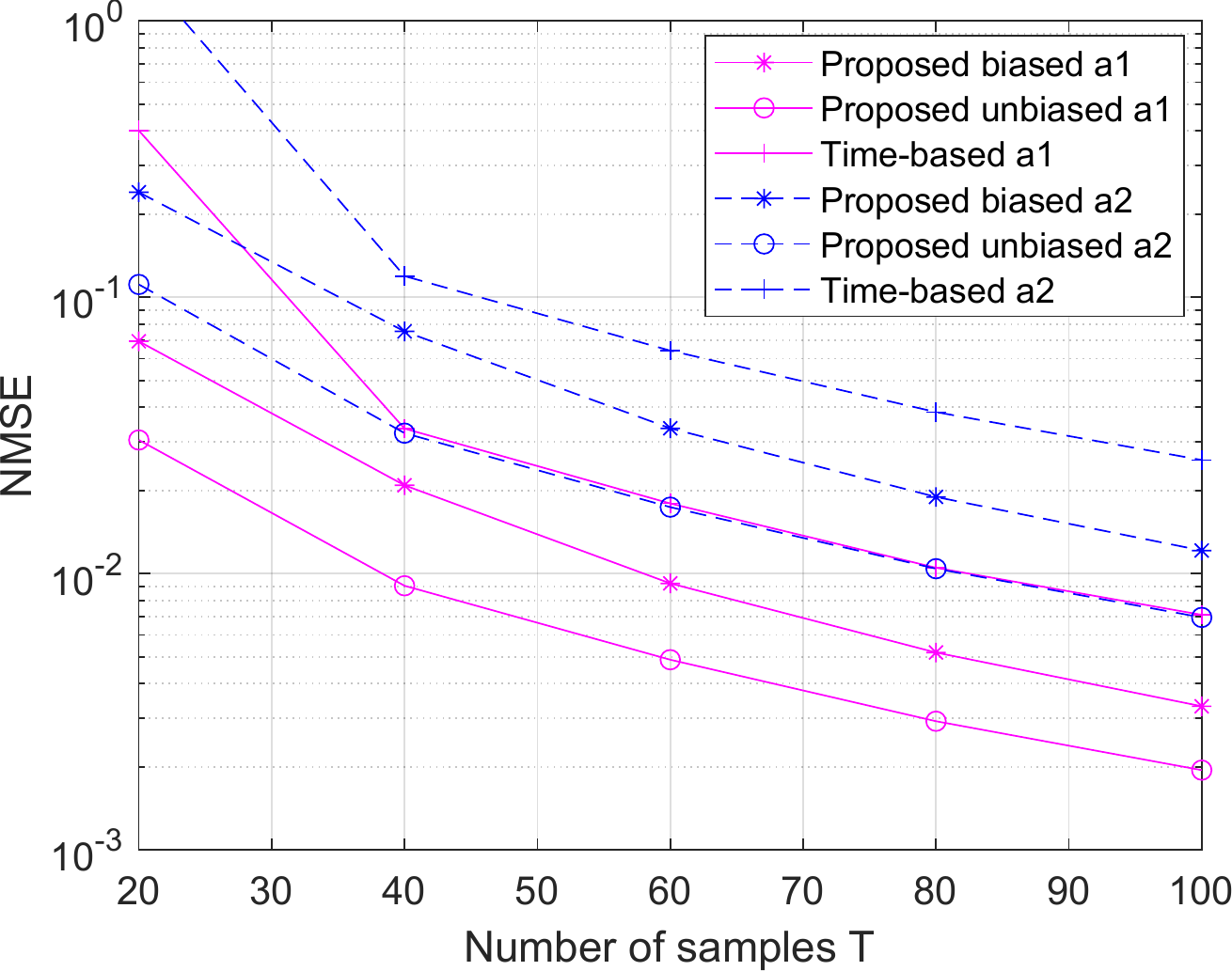}
		\vspace{-0.35cm}
	\caption{Normalized mean square errors of AR($2$) coefficient estimates versus $T=N_r$ with SNR$=0$ dB.}
\end{figure}

\subsection{Channel tracking}
In this subsection, we consider the AR($2$) Jakes' channel model with $a_1=1.8$ and $a_2=-0.9$. We apply a Kalman filter based channel estimation method from \cite{Kim'2021}. We assume $N_r=64$, the maximum size of the observation window is $T=N_r=64$ and the \ac{SNR} is equal to 0 dB. At each time $t$ the channel estimate is based on the $t$ concatenated observations of the received signal using the estimates of the AR($2$) coefficients based on those $t$ observations. The estimates of the AR($2$) coefficients are obtained using the same methods as in the previous section: the proposed approach and the time-based method. The genie method is referred to the channel estimation using the true values of the coefficients. The NMSEs of the different instantaneous (at time instant $t<T$ based on $t$ observations) channel estimation methods are compared in Figure~2. We observe that the proposed approach provides the best performance, especially for the unbiased case which is close to the one of the genie method. Moreover, we notice that a good performance is obtained with a quite small number of observations as compared to the size of the received signal.
\begin{figure}[t]\label{Instan_channel_simu}
		\center
	\includegraphics[width=7cm]{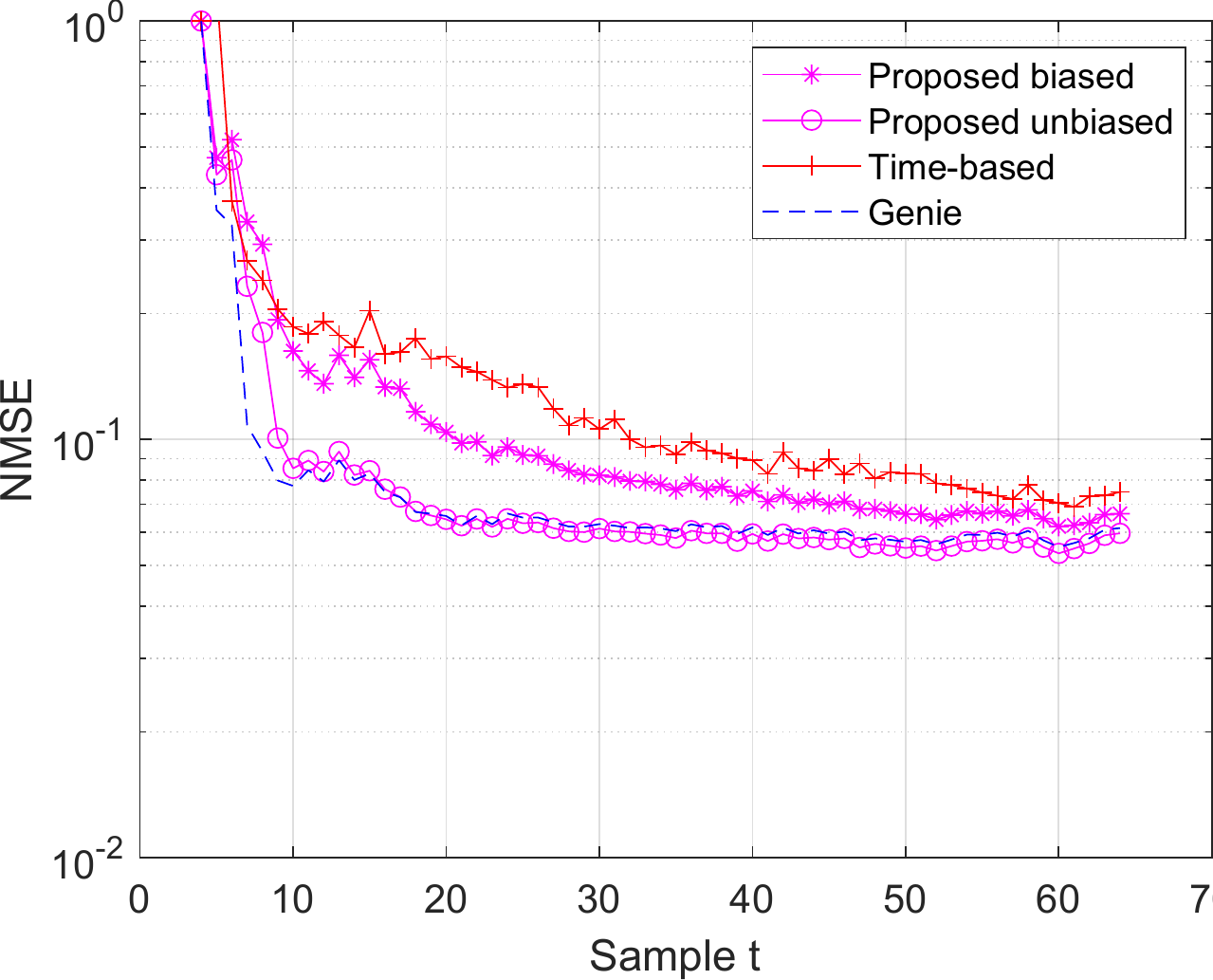}
	\vspace{-0.35cm}
	\caption{Normalized mean square errors of the instantaneous channel estimates at $t<T$ with SNR$=0$ dB.}
\end{figure}
The same channel estimation methods are compared in Figure~3 in terms of the NMSE for $T=N_r=64$ versus SNR. We observe that the proposed method is more beneficial especially at a lower SNR for which the observation noise is higher.
\begin{figure}[t]\label{SNR_simu}
	\center
	\includegraphics[width=7cm]{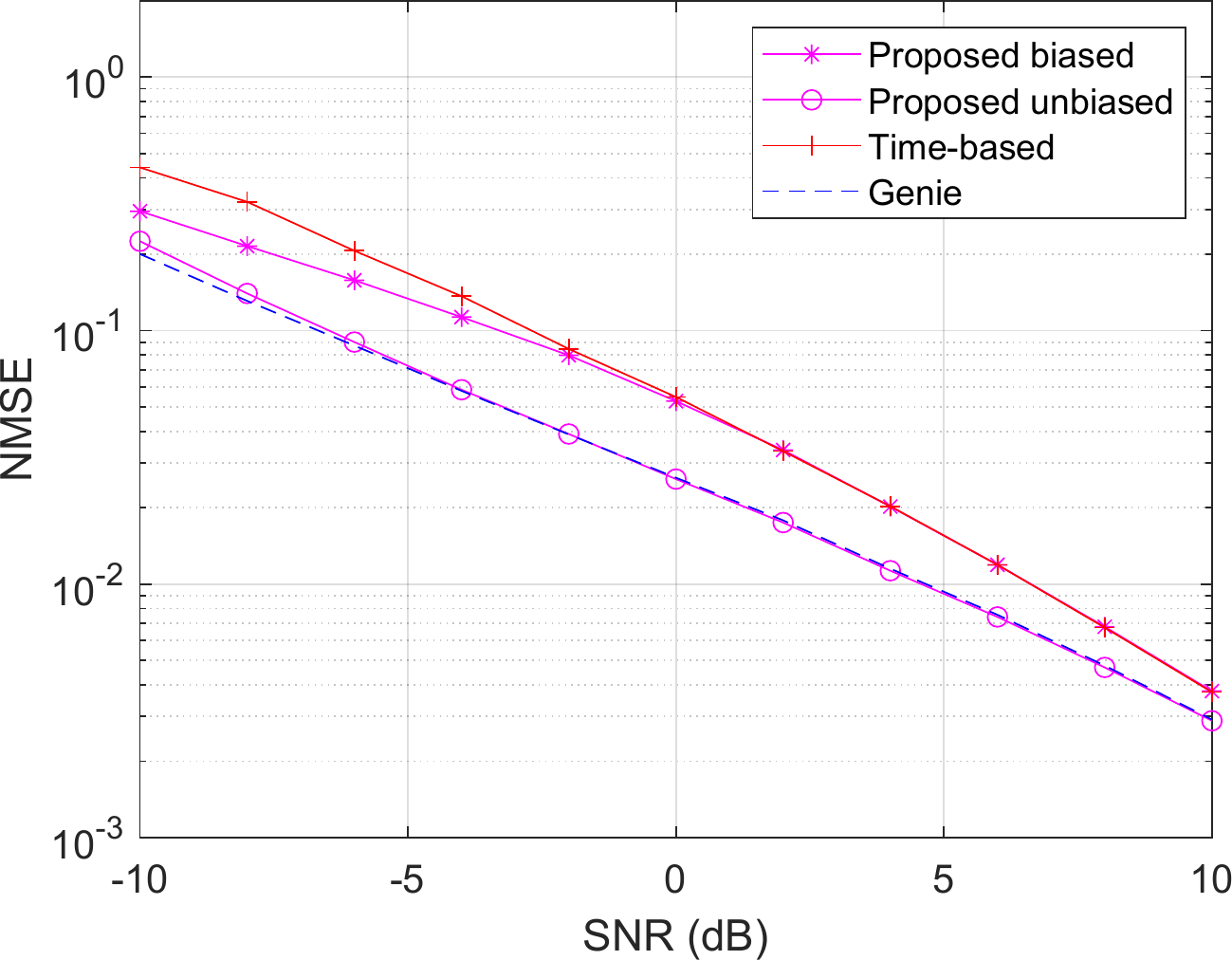}
		\vspace{-0.35cm}
	\caption{Normalized mean square errors of the channel estimates versus SNR.}

\end{figure}
Finally, the above methods are compared in Figure~4 in terms of the NMSE for SNR$=-5$ dB versus $N=T$. In this case we assume that the AR($2$) coefficient estimates are based on $T$ observations. Again, we observe that the proposed approach provides the best performance, close to the genie for the unbiased case.
\begin{figure}[t]\label{N_simu}
		\center
	\includegraphics[width=7cm]{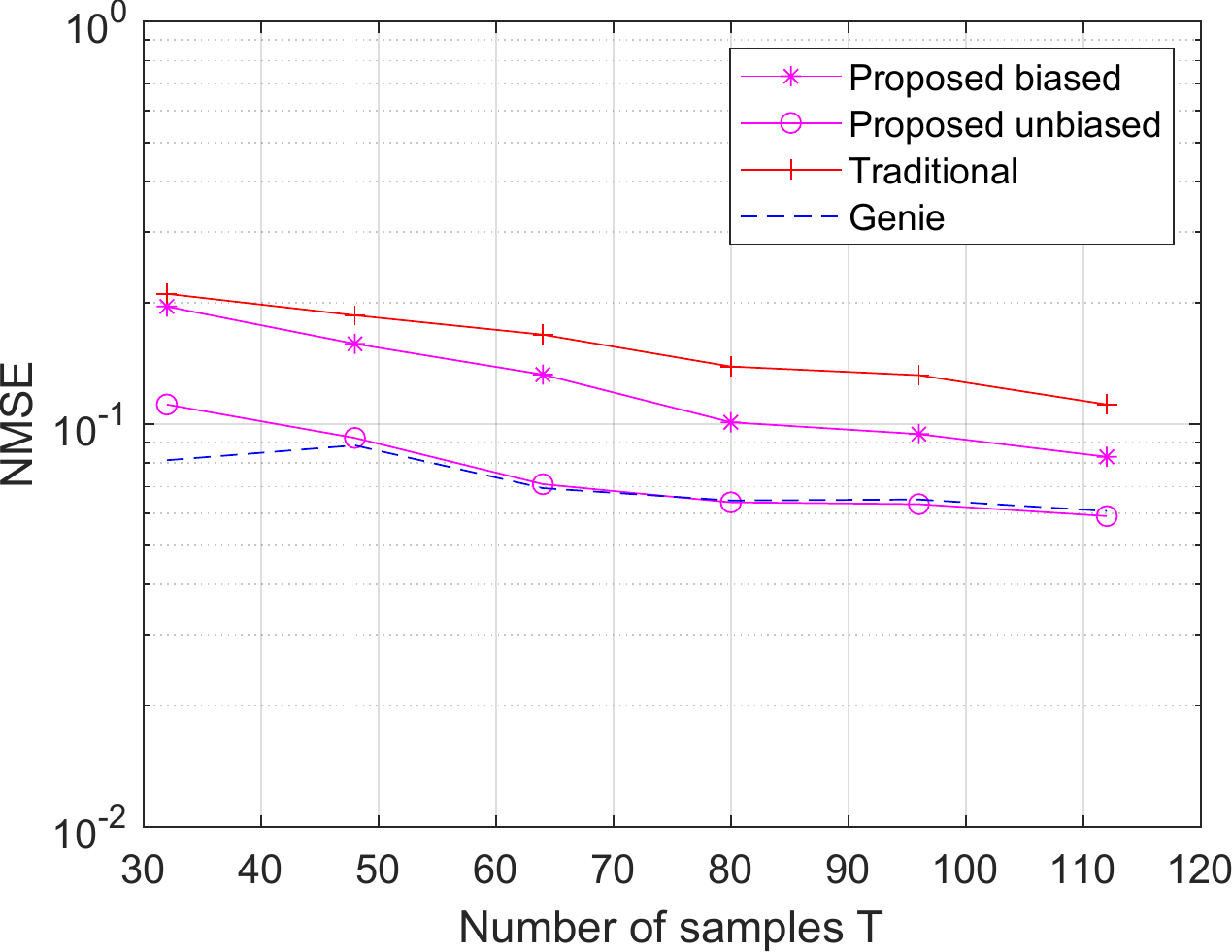}
	\vspace{-0.35cm}
	\caption{Normalized mean square errors of the channel estimates versus $T=N_r$ with SNR$=-5$ dB.}
\end{figure}

\section{Conclusions}
\label{Sec:Con}
In this paper, we considered the problem of estimating the parameters of an \ac{AR}
process, which can model the evolution of time-varying wireless channel in \ac{SIMO} systems.
This problem is motivated by the observation that when the parameters of the \ac{AR}
are properly set, the model can be used to develop channel estimators and predictors.
We have shown the almost sure convergence of the proposed estimate to the true value. The proposed estimates have been used for channel tracking for a specific case of Jakes' model. However, broader channel models following the \ac{AR} need to be further studied.

\section{Appendix}
\subsection{Proof of Lemma 1}\label{Proof_lemma1}
The random variable given by $\hat{r}^b(k)$ is integrable and has a finite variance $\sigma^2_{\hat{r}^b(k)}$ and a finite mean \linebreak $\E\left(\hat{r}^b(k)\right)=(1-|k|/T)r(k)$. Hence, from the Chebyshev's inequality, for any $\epsilon >0$, we have
\begin{align*}
\mathbb{P} \left[\left|\hat{r}^b(k) - \E\left(\hat{r}^b(k)\right) \right| \geq \epsilon \right] \leq \frac{\sigma^2_{\hat{r}^b(k)}}{\epsilon^2}.
\end{align*}
In the remaining of the proof, we consider the term $\sigma^2_{\hat{r}^b(k)}$ and calculate its upper bound.
\par We define $\widetilde{\bH} \triangleq \bZ\bQ$
where $\bZ \in \C^{N \times T}$ has i.i.d. elements $z_{n}(t) \sim {\cal CN}(0, 1)$ with rows denoted by $\bz_n$ and $\bQ=\sigma_x\bR^{1/2}+\sigma_w \bI_{T}=[{\bq}_0, \ldots, {\bq}_{T-1}]$ with $\bq_t\in \C^{T \times 1}$.
The entries of the matrices $\widehat{\bH}$ from \eqref{H_hat} and $\widetilde{\bH}$ have the same complex Gaussian distribution with independent rows and dependent columns with covariance matrices ${\bR}_{\hat{\bh}}$ defined in \eqref{Rhat}.
We can write, for $k=1-T, \ldots, T-1$:
\begin{align*}
\varr[\hat{r}^b(k)]
& = \varr\Big[\frac{1}{\sigma_x^2N_rT}\sum_{n=0}^{N_r-1}\sum_{t=0}^{T-1}\bz_n \bq_{t+k} \bq_{t}^{\sf T} \bz_n^{\sf H}\Big].
\end{align*}
Using the Cauchy-Schwarz inequality, after some steps we get
\begin{align*}
\varr[\hat{r}^b(k)]& \leq \frac{\norme{{\bR}_{\hat{\bh}}}}{\sigma_x^2N_r^2T^2}\sum_{n=0}^{N_r-1} \sum_{i=0}^{T-1} \varr|z_n(i)|^2 = \frac{\norme{{\bR}_{\hat{\bh}}}}{\sigma_x^2N_rT}
\end{align*}
where $z_n(i)$ are i.i.d. with a unit variance and $\norme{\bQ}^2=\norme{{\bR}_{\hat{\bh}}}$ is bounded as the norm of the $\bR$ is bounded because of the absolute summability assumption of the covariance coefficients $\sum_{k=1-T}^{T-1} r(k) < \infty$.
Hence, after some steps for any $\epsilon>0$
\begin{equation*}
\mathbb{P} \left[\left|\hat{r}^b(k) - {r}(k) \right| \geq \epsilon \right]
\leq \frac{\norme{{\bR}_{\hat{\bh}}}}{\epsilon^2\sigma_x^2N_rT}.
\end{equation*}
As $T,N_r$ converge to infinity, we get the almost sure convergence of the proposed estimator.
\par The unbiased case is proved following similar steps with the variance bounded by $\frac{\norme{{\bR}_{\hat{\bh}}}}{\epsilon^2\sigma_x^2N_r(T-|k|)}$.

\subsection{Proof of Lemma 2}\label{Proof_lemma2}
As in the above proof, the statistical behavior of the entries of the matrix $\widehat{\bH}$ is equivalent to the statistics of the entries of $\widetilde{\bH}$ defined in Appendix~\ref{Proof_lemma1}. We can apply Theorem~1 from \cite{Vino'2014} to get the almost sure convergence of $\widehat{\bR}^b$ and $\widehat{\bR}^u$ for both, biased and unbiased cases.
From the fact that $\bR$ and $\widehat{\bR}$ are Hermitian nonnegative Toeplitz (and so are $\bR_p$ and $\widehat{\bR}_p$), we have $\norme{\bR_p-\widehat{\bR}_p} \leq \norme{\bR-\widehat{\bR}}$ for any $p \in \left\{1,\ldots, T\right\}$ for both cases, we get the result.

\subsection{Proof of Theorem 1}\label{Proof_theorem1}
The proof is based on the usage of Lemma~1 and Lemma~2. Let $\widehat{\br}_p$ be the biased or unbiased vector of estimated covariance coefficients.
We can write
\begin{equation*}
\norme{{\ba}_p-\hat{\ba}_p} \leq \norme{{\bR}_p^{-1}}\norme{\br_p-\hat{\br}_p} + \norme{{\bR}_p^{-1}}\norme{\widehat{\bR}_p}\norme{\bI_{T}-{\bR}\widehat{\bR}^{-1}}.
\end{equation*}
From Lemma~1, we get the almost sure convergence to zero of the first term. The second term converges to zero from Lemma~2 and noticing that the norm $\norme{\widehat{\bR}_p}$ is bounded almost surely as $T,N_r \to \infty$.

\small
\bibliographystyle{ieeetr}
\bibliography{Bibliography}

\begin{thebibliography}{10}

\bibitem{Yan:01}
M.~Yan and D.~Rao, ``Performance of an array receiver with a {Kalman} channel
  predictor for fast {Rayleigh} flat fading environments,'' {\em IEEE Journal
  on Selected Areas in Communications}, vol.~6, no.~6, pp.~1164--1172, 2001.

\bibitem{Hijazi:10}
H.~Hijazi and L.~Ros, ``Joint data {QR}-detection and {Kalman} estimation for
  {OFDM} time-varying {R}ayleigh channel complex gains,'' {\em IEEE Trans. on
  Communications}, vol.~58, pp.~170--177, Jan. 2010.

\bibitem{Zhang:07B}
Y.~Zhang, S.~B. Gelfand, and M.~P. Fitz, ``Soft-output demodulation on
  frequency-selective {R}ayleigh fading channels uing {AR} channel models,''
  {\em IEEE Trans. on Communications}, vol.~55, pp.~1929--1939, Oct. 2007.

\bibitem{Lehmann:08}
F.~Lehmann, ``{Blind estimation and detection of space-time Trellis coded
  transmissions over the Rayleigh fading {MIMO} channel},'' {\em IEEE Trans. on
  Communications}, vol.~56, pp.~334--338, Mar. 2008.

\bibitem{Abeida:10}
H.~Abeida, ``Data-aided {SNR} estimation in time-variant {R}ayleigh fading
  channels,'' {\em IEEE Trans. on Signal Processing}, vol.~58, pp.~5496--5507,
  Nov. 2010.

\bibitem{Truong:13}
K.~T. {Truong} and R.~W. {Heath}, ``Effects of channel aging in massive {MIMO}
  systems,'' {\em Journal of Communications and Networks}, vol.~15, no.~4,
  pp.~338--351, 2013.

\bibitem{Kong:15}
C.~{Kong}, C.~{Zhong}, A.~K. {Papazafeiropoulos}, M.~{Matthaiou}, and
  Z.~{Zhang}, ``Sum-rate and power scaling of massive {MIMO} systems with
  channel aging,'' {\em IEEE Trans. on Communications}, vol.~63, no.~12,
  pp.~4879--4893, 2015.

\bibitem{Zhang:16}
C.~{Zhang}, D.~{Guo}, and P.~{Fan}, ``Tracking angles of departure and arrival
  in a mobile millimeter wave channel,'' in {\em IEEE ICC}, pp.~1--6, 2016.

\bibitem{Papa:18}
A.~{Papazafeiropoulos} and T.~{Ratnarajah}, ``Modeling and performance of
  uplink cache-enabled massive {MIMO} heterogeneous networks,'' {\em IEEE
  Trans. on Wireless Communications}, vol.~17, no.~12, pp.~8136--8149, 2018.

\bibitem{Kim:20}
H.~{Kim}, S.~{Kim}, H.~{Lee}, C.~{Jang}, Y.~{Choi}, and J.~{Choi}, ``Massive
  {MIMO} channel prediction: {Kalman} filtering vs. machine learning,'' {\em
  IEEE Trans. on Communications}, pp.~1--1, 2020.
\newblock early access.

\bibitem{Yuan:20}
J.~{Yuan}, H.~Q. {Ngo}, and M.~{Matthaiou}, ``Machine learning-based channel
  prediction in massive {MIMO} with channel aging,'' {\em IEEE Trans. on
  Wireless Communications}, vol.~19, no.~5, pp.~2960--2973, 2020.

\bibitem{Esfandiari:20}
M.~Esfandiari, S.~A. Vorobyov, and M.~Karimi, ``New estimation methods for
  autoregressive process in the presence of white observation noise,'' {\em
  Signal Processing (Elsevier)}, vol.~2020, no.~171, pp.~10780--10790, 2020.

\bibitem{Gray'06}
R.~M. Gray, {\em Toeplitz and circulant matrices: A review}.
\newblock Now Pub., 2006.

\bibitem{Kay'1988}
S.~M. Kay, {\em Modern spectral estimation: {Theory} and application}.
\newblock Englewood Cliffs, N.J.: Prentice-Hall, 1988.

\bibitem{Vino'2014}
J.~Vinogradova, R.~Couillet, and W.~Hachem, ``Estimation of {T}oeplitz
  covariance matrices in large dimensional regime with application to source
  detection,'' {\em IEEE Trans. on Signal Processing}, vol.~63, 03 2014.

\bibitem{Baddour:05}
K.~E. {Baddour} and N.~C. {Beaulieu}, ``Autoregressive modeling for fading
  channel simulation,'' {\em IEEE Trans. on Wireless Communications}, vol.~4,
  no.~4, pp.~1650--1662, 2005.

\bibitem{ElHusseini'2019}
A.~H. {El Husseini}, E.~P. Simon, and L.~Ros, ``{Second-order autoregressive
  model-based Kalman filter for the estimation of a slow fading channel
  described by the Clarke model: Optimal tuning and interpretation},'' {\em DSP
  (Elsevier)}, vol.~90, pp.~125--141, 2019.

\bibitem{Kim'2021}
H.~{Kim}, S.~{Kim}, H.~{Lee}, C.~{Jang}, Y.~{Choi}, and J.~{Choi}, ``Massive
  {MIMO} channel prediction: {K}alman filtering vs. machine learning,'' {\em
  IEEE Trans. on Communications}, vol.~69, no.~1, pp.~518--528, 2021.

\end{thebibliography}

\end{document}